\documentclass[twoside,onecolumn]{article}

\usepackage[english]{babel} 

\usepackage[hmarginratio=1:1,top=32mm,columnsep=20pt]{geometry} 
\usepackage[hang, small,labelfont=bf,up,textfont=it,up]{caption} 
\usepackage{booktabs} 

\usepackage{geometry} 
\usepackage{graphicx} 

\usepackage{booktabs} 
\usepackage{array} 
\usepackage{paralist} 
\usepackage{verbatim} 
\usepackage{subfig} 
\usepackage{amsmath}	
\allowdisplaybreaks[1]	
\numberwithin{equation}{section}	
\usepackage{amsfonts}	
\usepackage{amsthm}	
\usepackage{cases}
\usepackage{bbm}		
\usepackage{braket}	
\usepackage{mathtools}	
\usepackage[english]{varioref}	
\usepackage[multiple]{footmisc}	

\usepackage{fancyhdr} 
\theoremstyle{plain}

\theoremstyle{plain}
\newtheorem{Remark}{Remark}[section]
\theoremstyle{plain}
\newtheorem{Proposition}{Proposition}[section]
\theoremstyle{plain}
\newtheorem{Lemma}{Lemma}[section]
\theoremstyle{plain}
\newtheorem{Theorem}{Theorem}[section]
\theoremstyle{plain}

\theoremstyle{plain}
\newtheorem{Definition}{Definition}[section]
\theoremstyle{assumption}

\DeclareMathOperator*{\var}{\text{var}}

\usepackage{sectsty}
\allsectionsfont{\sffamily\mdseries\upshape} 

\usepackage{abstract} 

\usepackage{titlesec} 
\titleformat{\section}[block]{\large\bfseries\centering}{\thesection.}{1em}{} 
\titleformat{\subsection}[block]{\large\bfseries}{\thesubsection.}{1em}{} 

\usepackage[nottoc,notlof,notlot]{tocbibind} 
\usepackage[titles,subfigure]{tocloft} 

\usepackage[a-1b]{pdfx} 
\usepackage{titling} 
\usepackage{hyperref}

\hypersetup{%
    pdfpagemode={UseOutlines},
    bookmarksopen,
    pdfstartview={FitH},
    colorlinks,
    linkcolor={black},
    citecolor={black},
    urlcolor={black}
  }

\pretitle{\begin{center}\Huge\bfseries} 
\posttitle{\end{center}} 
\title{Optimal reinsurance problem under fixed cost and exponential preferences} 
\author{%
\textsc{Brachetta M.}\thanks{Department of Mathematics, Politecnico of Milan, Piazza Leonardo da Vinci, 32, 20133 Milano, Italy.} \thanks{Corresponding author.} \\[1ex]
\normalsize \href{mailto:matteo.brachetta@unich.it}{matteo.brachetta@unich.it}
\and
\textsc{Ceci, C.}\thanks{Department of Economics, University of Chieti-Pescara, Viale Pindaro, 42 - 65127 Pescara, Italy.} \\[1ex]
\normalsize \href{mailto:c.ceci@unich.it}{c.ceci@unich.it}
}
\date{} 

\providecommand{\keywords}[1]{\textbf{\textit{Keywords:}} #1}

\providecommand{\msccodes}[1]{\textbf{\textit{MSC Classification codes:}} #1}

\begin{document}
\maketitle

\begin{abstract}
\noindent 
We investigate an optimal reinsurance problem for an insurance company facing a constant fixed cost when the reinsurance contract is signed. The insurer needs to optimally choose both the starting time of the reinsurance contract and the retention level in order to maximize the expected utility of terminal wealth.
This leads to a mixed optimal control/optimal stopping time problem, which is solved by a two-step procedure: first considering the pure-reinsurance stochastic control problem and next discussing a time-inhomogeneous optimal stopping problem with discontinuous reward. Using the classical Cram\'er-Lundberg approximation risk model, we prove that the optimal strategy  is deterministic and depends on the model parameters. In particular, we show that there exists a maximum fixed cost that the insurer is willing to pay for the contract activation. Finally, we provide some economical interpretations and numerical simulations.
\end{abstract}

\noindent\keywords{Optimal Reinsurance; Mixed Control Problem; Optimal Stopping; Transaction Cost.}\\
\noindent\msccodes{93E20, 91B30, 60G40, 60J60.}\\



\section{Introduction}

Insurance business requires the transfer of risks from the policyholders to the insurer, who receives a risk premium as a reward. In some cases, it could be convenient to cede these risks to a third party, which is the reinsurance company. From the operational viewpoint, a risk-sharing agreement helps the insurer reducing unexpected losses, stabilizing operating results, increasing business capacity and so on. By means of a reinsurance treaty, the reinsurance company agrees to idemnify the primary insurer (cedent) against all or part of the losses which may occur under policies which the latter issued. The cedent will pay a reinsurance premium in exchange for this service. Roughly speaking, this is an insurance for insurers. When subscribing a reinsurance treaty, a natural question is to determine the (optimal) level of the retained losses.
Optimal reinsurance problems have been intensively studied by many authors under different criteria, especially through expected utility maximization and ruin probability minimization, see for example \cite{definetti:1940}, \cite{buhlmann:1970}, \cite{gerber69}, \cite{irgens_paulsen:optcontrol}, \cite{BC:IME2019} and references therein. \\

The main novelty of this article is that subscription costs are considered. In practice, when the agreement is signed, a fixed cost is usually paid in addition to the reinsurance premium. This aspect has not been investigated by nearly all the studies, except for \cite{egami} and \cite{lietal2015}. 
In the former work the authors discussed  the reinsurance problem subject to a fixed cost for buying reinsurance and a time delay in completing the reinsurance transaction. They solved the problem considering a performance criterion with linear current reward and showed that it is optimal to buy reinsurance when the surplus lies in a bounded interval depending on the delay time. In the latter paper, under the criterion of minimizing the ruin probability, the original problem is reduced to a time-homogeneous optimal stopping problem. In particular, the authors show that the fixed cost forces the insurer to postpone buying reinsurance until the surplus process hits a certain level.

Hence the presence of a fixed cost is closely related to the possibility of postponing the subscription of the reinsurance agreement. This, in turn, involves an optimal stopping problem, which is attached to the optimal choice of the retention level, which is a well known stochastic control problem.
The novelty of our paper consists in considering this mixed stochastic control problem under the criterion of maximizing the expected utility of terminal wealth. The strategy of the insurance company consists of the retention level of a proportional reinsurance and the subscription timing. When the contract is signed, a given fixed cost is paid and the optimal retention level is applied. We use a diffusion approximation to model the insurer's surplus process (see \cite{schmidli:2018risk}).
The insurance company has exponential preferences and is allowed to invest in a risk-less bond.

As already mentioned, this setup leads to a combined problem of optimal stopping and stochastic control with finite horizon,  which we will solve by a two-step procedure. First, we provide the solution of the pure reinsurance problem (with starting time equal to zero). Next, we discuss an optimal stopping time problem with a suitable reward function depending on the value function of the pure reinsurance problem. 
Differently to \cite{egami} and \cite{lietal2015}, the associated optimal stopping problem turns out to be time-inhomogeneous and with discontinuous stopping reward with respect to the time. We provide an explicit solution, also showing  that the optimal stopping time is deterministic. Moreover, we find that only two cases possible, depending on the model parameters. When the fixed cost is greater than a suitable threshold (whose analytical expression is available), the optimal choice is not to subscribe the reinsurance; otherwise,  the insurer immediately subscribes the contract.

The paper is organized as follows. In Section \ref{section:formulation}, we describe the model and formulate the problem as a mixed stochastic control problem, that is a problem which involves both optimal control and stopping.  In Section \ref{sec:reins} we discuss the pure reinsurance problem (without stopping) by solving  the associated Halmilton-Jacobi-Bellman equation. Section \ref{sec:red} is devoted to the reduction of the original (mixed) problem to a suitable optimal stopping problem, which is then investigated in Section \ref{sec:optstop}. Here we provide a Verification Theorem and we solve the associated variational inequality. In Section \ref{sec:solution} we give the explicit solution to the original problem and we discuss some economic implications of our results. Finally, in Section \ref{sec:num} some numerical simulations are performed in order to better understand the economic interpretation of our findings.


\section{Problem formulation}
\label{section:formulation}

\subsection{Model formulation}

Let $T>0$ be a finite time horizon and assume that $(\Omega,\mathcal{F},\mathbb{P},\mathbb{F})$ is a complete probability space endowed with a filtration $\mathbb{F}\doteq\{\mathcal{F}_t\}_{t\in [0,T]}$ satisfying the usual conditions. 

Let us denote by $R=\{R_t\}_{t\in [0,T]}$ the surplus process of an insurance company. 
There is a wide range of risk models in the actuarial literature, see for instance \cite{grandell:risk} and \cite{schmidli:2018risk}. In the Cram\'er-Lundberg risk model the claims arrival times are described by the sequence of claims arrival times $\{T_n\}_{n\geq1}$, with $T_n<T_{n+1}$ $\mathbb{P}$-a.e. $\forall n\ge1$, while the corresponding claim sizes are given by $\{Z_n\}_{n\geq1}$. In particular, the number of occurred claims up to time $t\ge0$ is equal to
\[
N_t =  \sum_{n=1} \mathbbm{1}_{\{T_n \leq t\}},
\]
and it is assumed to be a Poisson process with constant intensity $\lambda >0$, independent of the sequence $\{Z_n\}_{n\geq1}$. Moreover, $\{Z_n\}_{n\geq1}$ are i.i.d. random variables with common probability distribution function $F_Z(z)$, $z \in (0, + \infty)$, having finite first and second moments denoted by $\mu>0$ and $\mu_2>0$, respectively. In this context the surplus process is given by
\begin{equation}
\label{eqn:CLmodel}
  R_0 + ct - \sum_{n=1}^{N_t}Z_n , \qquad R_0>0 ,
\end{equation}
where $R_0$ is the initial capital and $c>0$ denotes the gross risk premium rate. We can show that for any $t\ge0$
\[
\mathbb{E}\biggl[ \sum_{n=1}^{N_t}Z_n \biggr] = \lambda\mu t
\qquad \text{and} \qquad 
\var\biggl[ \sum_{n=1}^{N_t}Z_n \biggr] = \lambda\mu_2 t .
\]

In this paper we use the diffusion approximation of the Cram\'er-Lundberg model \eqref{eqn:CLmodel}, see for example \cite{grandell:risk}. Precisely, we assume that the surplus process follows this stochastic differential equation (SDE):
\begin{equation*}
dR_t = p\,dt + \sigma_0 \, dW_t, \qquad R_0>0,
\end{equation*}
where $W=\{W_t\}_{t\in [0,T]}$ is a standard Brownian motion, $\sigma_0=\sqrt{\lambda\mu_2}$ and $p$ denotes the insurer's net profit, that is $p=c-\mu\lambda$. In particular, under the expected value principle we have that $c=(1+\theta_i)\mu\lambda$ and hence $p=\theta_i\mu\lambda$, with $\theta_i>0$ representing the insurer's safety loading.\\

We allow the insurer to invest her surplus in a risk-free asset with constant interest rate $R>0$:
\[
dB_t = B_t R dt, \quad B_0=1 , 
\]
hence the wealth process $X=\{X_t\}_{t\in [0,T]}$ evolves according to
\begin{equation}\label{W1}
dX_t= R X_t dt + p\,dt + \sigma_0 \, dW_t, \qquad X_0 =R_0>0.
\end{equation}
The explicit solution of the SDE \eqref{W1} is given by the following equation: 
\begin{equation}\label{X}
X_t = R_0 e^{Rt} + \int_0^t e^{R(t-s)} p ds + \int_0^t e^{R(t-s)} \sigma_0 dW_s, \qquad t\in[0,T].
\end{equation}

Now let $\tau$ denote an $\mathbb{F}$-stopping time. At time $\tau$ the insurer can subscribe a proportional reinsurance contract with retention level $u\in[0,1]$, transferring part of her risks to the reinsurer. More precisely, $u$ represents the percentage of retained losses, so that $u=0$ means full reinsurance, while $u=1$ is equivalent to no reinsurance. In order to buy a reinsurance agreement, the primary insurer pays a reinsurance premium $q(u) \geq 0$. When the reinsurance contract is signed at time $t=0$, the Cram\'er-Lundberg risk model \eqref{eqn:CLmodel} is replaced by the following equation:
\[
  R_0 + (c - q(u)) t - \sum_{n=1}^{N_t} u Z_n , \qquad R_0>0.
\]

Under the expected value principle we have that $q(u) = (1+\theta)(1-u)\mu\lambda$, $u \in [0,1]$, with the reinsurer's safety loading $\theta $ satistying $\theta> \theta_i$ (preventing the insurer from gaining a risk-free profit).

Let us denote by $R^u=\{R^u_t\}_{t\in [0,T]}$ the reserve process in the Cram\'er-Lundberg approximation associated with a given reinsurance strategy $\{u_t\}_{t\in [0,T]}$ when the reinsurance contract is signed at time $t=0$. Following \cite{eisenbergschmidli:2009},  under the expected value principle,  $R^u$  follows
\begin{equation}
\label{eqn:R^u}
dR^u_t = (p-q + qu_t)\,dt + \sigma_0 u_t \,dW_t , \qquad R^u_0=R_0,
\end{equation}
where $q = \theta\lambda\mu$ denotes the reinsurer's net profit. We set $q>p$ (non-cheap reinsurance).
The wealth process under the strategy $\{u_t\}_{t\in [0,T]}$ evolves according to this SDE:
\begin{equation}\label{W2}
dX^u_t= R X^u_t dt + dR^u_t, \qquad X^u_0 = R_0,
\end{equation}
which admits this explicit representation:
\begin{equation}\label{Xu}
X^u_t = R_0e^{Rt} + \int_0^t e^{R(t-s)}(p-q + qu_s)\,ds
+\int_0^t e^{R(t-s)} \sigma_0 u_s \,dW_s .
\end{equation}

We assume that a constant fixed cost $K>0$ is paid when the reinsurance contract is subscribed. The insurer decides when the reinsurance contract starts and which retention level is applied. Hence the insurer's strategy is a couple $\alpha = ( \tau, \{u_t\}_{t\in [\tau,T]})$, with $\tau \leq T$.  Let $H_t = I_{\{\tau \leq t\}}$ be the indicator process of the contract starting time. For $\tau < T$ $\mathbb{P}$-a.s. , the total wealth $X^\alpha = \{X^\alpha_t\}_{t\in [0,T]}$ associated with a given strategy $\alpha$ is given by 

\begin{equation}\label{W}
dX^\alpha_t = (1- H_t) dX_t + H_t dX^u_t - K dH_t, \qquad  X^\alpha_0 =R_0 >0,
\end{equation}
while on the event  $\{ \tau = T \}$ we have that 
\begin{equation}\label{W0}
dX^\alpha_t = dX_t,  \qquad  X^\alpha_0 =R_0 >0, 
\end{equation}
where $X$ satisfies equation \eqref{W1}.

Equation \eqref{W} can be written more explicitly as
\begin{equation}
\label{wealth_alpha}
dX^\alpha_t =
\begin{cases}
	dX_t,  & t  < \tau, \quad  X_0 = R_0,
	\\
	dX^u_t, & \tau < t \leq T, \quad X^u_\tau =  X_\tau - K, 
\end{cases}
\end{equation}
where $X$ and $X^u$ satisfy equations \eqref{W1} and \eqref{W2}, respectively.

In our setting the null reinsurance corresponds to the choice $\tau = T$, $\mathbb{P}$-a.s., to which we associate the strategy $\alpha_\text{null} = ( T, 1 )$ and
\[
X^{\alpha_\text{null} }_t = X_t, \qquad t \in [0,T].
\]

\subsection{The utility maximization problem}

The insurers' objective is to maximize the expected utility of the terminal wealth:
\begin{equation}
\label{eqn:max_pb}
\sup_{\alpha\in\mathcal{A}}{\mathbb{E}\bigl[ U(X_T^{\alpha}) \bigr]},
\end{equation}
where $U:\mathbb{R}\to[0,+\infty)$ is the utility function representing the insurer's preferences and 
$\mathcal{A}$ the class of admissible strategies (see Definition \ref{def_A} below). 

We focus on CARA (\textit{Constant Absolute Risk Aversion}) utility functions, whose general expression is given by
$$
U(x) = 1-e^{-\eta x}, \qquad x\in\mathbb{R},
$$
where $\eta>0$ is the risk-aversion parameter. This utility function is highly relevant in economic science and particularly in insurance theory. Indeed, it is commonly used for reinsurance problems (e.g. see \cite{BC:IME2019} and references therein).\\
 
The optimization problem is a mixed optimal control problem. That is,  the insurer's controls involve the timing of the reinsurance contract subscription and the retention level to apply. 

\begin{Definition}[Admissible strategies] \label{def_A}
We denote by $\mathcal{A}$ the set of admissible strategies  $\alpha = ( \tau, \{u_t\}_{t \in [\tau,T]} )$, where
$\tau$ is an $\mathbb{F}$-stopping time such that $\tau \leq T$ and $\{u_t\}_{t \in [\tau,T]}$ is  an $\mathbb{F}$-predictable process with values in $[0,1]$. Let us observe that the  null strategy $\alpha_\text{null} = ( T, 1 )$ is included in $\mathcal{A}$. When we want to restrict the controls to the time interval $[t,T]$, we will use the notation $\mathcal{A}_t$.
\end{Definition}

\begin{Proposition}
Let $\alpha\in\mathcal{A}$, then  
$$\mathbb{E}\bigl[ e^{-\eta X^{\alpha}_T} \bigr] < +\infty. $$
\end{Proposition}
\begin{proof} 
Using equations \eqref{Xu} and \eqref{wealth_alpha}, we have that
\[
\begin{split}
\mathbb{E}\bigl[ e^{-\eta X^{\alpha}_T} \bigr]  & = \mathbb{E}\bigl[ e^{-\eta X_T} I_{\{\tau = T\}} \bigr]  \\
&+ \mathbb{E}\bigl[ e^{-\eta (X_\tau - K) e^{R(T- \tau)} } 
e^{- \eta \int_\tau^T e^{R(T-s)} (p-q + q u_s) ds} e^{- \eta \int_\tau^T e^{R(T-s)} \sigma_0 u_s dW_s} I_{\{\tau < T\}} \bigr]  .
\end{split}
\]
Taking into account the expression \eqref{X} we get 
\[
\begin{split}
\mathbb{E}\bigl[ e^{-\eta X_T} I_{\{\tau = T\}} \bigr] 
&\leq \mathbb{E}\bigl[  e^{-\eta R_0 e^{R T}} 
e^{- \eta \int_0^T e^{R(t-s)} p ds } e^{- \eta \int_0^T  e^{R(t-s)} \sigma_0 dW_s} \bigr]  \\
&\leq 	\mathbb{E}\bigl[ e^{- \eta \int_0^T  e^{R(t-s)} \sigma_0 dW_s} \bigr]  \\
&=		e^{ \frac{\eta^2}{2} \int_0^T  e^{2R(t-s)} \sigma^2_0 ds} < +\infty,
 \end{split}
\]
and denoting by $C$ a generic constant (possibly different from each line to another)
\[
\begin{split}
&\mathbb{E}\bigl[ e^{-\eta (X_\tau - K) e^{R(T- \tau)} } 
e^{- \eta \int_\tau^T e^{R(T-s)} (p-q + q u_s) ds} e^{- \eta \int_\tau^T e^{R(T-s)} \sigma_0 u_s dW_s} I_{\{\tau < T\}} \bigr] \\
&\leq C\times  \mathbb{E}\bigl[ e^{-\eta (X_\tau - K) e^{R(T- \tau)} } e^{- \eta \int_\tau^T e^{R(T-s)} \sigma_0 u_s dW_s} I_{\{\tau < T\}} \bigr] \\
& \leq C \times  \Big(  \mathbb{E}\bigl[ e^{- 2 \eta (X_\tau - K) e^{2R(T- \tau)} } I_{\{\tau < T\}}  \bigr]  +   \mathbb{E}\bigl[  e^{- 2 \eta \int_\tau^T e^{R(T-s)} \sigma_0 u_s dW_s} I_{\{\tau < T\}} \bigr] \Big) \\
&\leq C \times  \Big(  \mathbb{E}\bigl[ e^{- 2 \eta X_\tau e^{2R(T- \tau)} } \bigr] + \mathbb{E}\bigl[  e^{- 2 \eta \int_\tau^T e^{R(T-s)} \sigma_0 u_s dW_s} I_{\{\tau < T\}} \bigr]  \Big) \\
&\leq C \times  \Big(  \mathbb{E}\bigl[ e^{- 2 \eta e^{2R(T- \tau)} \int_0^\tau  e^{R(t-s)} \sigma_0 dW_s} \bigr]  + 
\mathbb{E}\bigl[  e^{2 \eta^2 \int_\tau^T e^{2R(T-s)} \sigma^2_0 u^2_s I_{\{\tau < T\}}ds} \bigr] \Big) \\
&\leq C \times  \Big( e^{2 \eta^2 e^{4RT} \int_0^T e^{2R(t-s)} \sigma^2_0 ds}  + e^{2 \eta^2 \int_0^T e^{2R(T-s)} \sigma^2_0 ds} \Big)  < + \infty.
\end{split}
\]
\end{proof}


Let us introduce the value function associated to our problem \eqref{eqn:max_pb}:
\begin{equation}  \label{V} 
V(t,x) = \inf_{\alpha\in\mathcal{A}_t} { \mathbb{E} \bigl[e^{-\eta X^{\alpha, t,x} _T } \bigr]},
 \qquad (t,x)\in [0,T]\times\mathbb{R},
\end{equation}
where $X^{\alpha, t,x} = \{ X^{\alpha, t,x}_s \}_{s\in [t,T]}$  denotes the wealth given in equation \eqref{wealth_alpha}  with initial condition 
$(t,x) \in [0,T] \times \mathbb{R}$,  that is $X^{\alpha, t,x}_t = x$. We notice that
 \begin{equation}  \label{VT} 
 V(T,x) = e^{-\eta x} \qquad \forall x \in\mathbb{R},
\end{equation}
because $X^{\alpha,T,x} _T =x$ $\forall \alpha\in\mathcal{A}$.
 
 
\section{The pure reinsurance problem}
\label{sec:reins}

In order to have a self-contained article, in this section we briefly investigate a pure reinsurance problem, which corresponds to the problem \eqref{eqn:max_pb} with fixed starting time $t=0$. Precisely, we deal with
\[
\inf_{u\in\mathcal{U}}\mathbb{E}\bigl[e^{-\eta X^u_T}\bigr] ,
\]
where $\mathcal{U}$ denotes the class of admissible strategies  $u=  \{u_t\}_{t\in [0,T]}$, which are all the $\mathbb{F}$-predictable processes with values in $[0,1]$. 
Let us denote by $\bar V(t,x)$ the value function associated to this problem, that is
\begin{equation}\label{V bar}
\bar V(t,x) = \inf_{u\in\mathcal{U}_t}{ \mathbb{E}\bigl[e^{-\eta X^{u, t,x} _T} \bigr]},  \quad  (t,x) \in [0,T] \times \mathbb{R}, 
\end{equation}
with $\mathcal{U}_t$ denoting the restriction of $\mathcal{U}$ to the time interval $[t,T]$ and $\{X^{u, t,x}_s\}_{s \in [t,T]}$ denotes the process satisying equation \eqref{W2} with initial data $(t,x) \in [0,T] \times \mathbb{R}$.
It is well known that the value function \eqref{V bar} can be characterized as a classical solution to the associated Hamilton-Jacobi-Bellman (HJB) equation:
\begin{equation}
\begin{cases}
\min_{u\in[0,1]}\mathcal{L}^u \bar V(t,x) = 0, &\forall (t,x)\in[0,T)\times\mathbb{R} \\
\bar V(T,x) = e^{-\eta x } & \forall x\in\mathbb{R} ,
\end{cases}
\end{equation}
where, using equations \eqref{eqn:R^u} and \eqref{W2}, the generator of the Markov process $X^u$ is given by 
\[
\mathcal{L}^u f(t,x)  = \frac{\partial{f}}{\partial{t}}(t,x) 
+ (Rx + p - q + qu) \frac{\partial{f}}{\partial{x}}(t,x) 
+ \frac{1}{2}\sigma_0^2u^2 \frac{\partial^2{f}}{\partial{x^2}}(t,x) ,
\quad f\in\mathcal{C}^{1,2}((0,T) \times \mathbb{R}).  
\]
Under the ansatz $\bar V(t,x) = e^{-\eta x e^{R(T-t)} } \phi(t)$, the HJB equation reads as
\[
\phi'(t) + \Psi(t) \phi(t) = 0, \qquad \phi(T)=1 ,
\]
where
\[
\Psi(t) = \min_{u\in[0,1]}{\{ - \eta e^{R(T-t)} (p - q + qu)
+ \frac{1}{2}\sigma_0^2 u^2 \eta^2 e^{2R(T-t)} \}} .
\]
Solving the minimization problem we find the unique minimizer:
\[
u^*(t) = \frac{q}{\eta\sigma_0^2}e^{-R(T-t)} \lor 1, \qquad t \in [0,T] .
\]

Under the additional condition 
\begin{equation}
\label{eqn:qsmall}
q<\eta\sigma_0^2, 
\end{equation}
$u^*(t)$ simplifies to
\begin{equation} \label{u_star} 
u^*(t) = \frac{q}{\eta\sigma_0^2}e^{-R(T-t)} \in (0,1), \qquad t \in [0,T].
\end{equation}
Using this expression we readily obtain that
\begin{equation}
\label{eqn:Psi}
\Psi(t) = \eta e^{R(T-t)}(q-p) - \frac{1}{2}\frac{q^2}{\sigma_0^2} .
\end{equation}
By classical verification arguments, we can verify that the value function given in \eqref{V bar} takes this form:
\begin{equation}\label{ex}
	\begin{split}
\bar V(t,x) &= e^{-\eta x e^{R(T-t)} } e^{\int_t^T \Psi(s) ds} \\
&= e^{-\eta x e^{R(T-t)} } e^{\frac{\eta(q-p)}{R} (e^{R(T-t)}-1) } 
e^{- \frac{1}{2}\frac{q^2}{\sigma_0^2}(T-t)},
	\end{split}
\end{equation}
and, under the condition \eqref{eqn:qsmall},  equation \eqref{u_star} provides an optimal reinsurance strategy.

\begin{Remark}\label{idea}
Comparing the optimal strategy $u^*(s)$, $s \in [t,T]$, to the null reinsurance $u(s)=1$, $s \in [t,T]$, by means of \eqref{V bar} we get that
\begin{equation}
\label{eqn:V<g}
 \bar V(t, x) \leq  \mathbb{E}\bigl[ e^{-\eta X^{t,x}_T}  \bigr] \qquad \forall (t,x) \in [0,T] \times \mathbb{R}.
\end{equation} 
Moreover, by equation \eqref{W1} we have that 
\begin{equation}
\label{g2}
\begin{split}
g(t,x) \doteq \mathbb{E}\bigl[ e^{-\eta X^{t,x}_T}  \bigr] 
&= \mathbb{E}\bigl[  e^{-\eta x e^{R (T-t)}}
e^{- \eta \int_t^T e^{R(T-s)} p ds } e^{- \eta \int_t^T  e^{R(T-s)} \sigma_0 dW_s} \bigr] \\
&= e^{-\eta x e^{R (T-t)}}  e^{-\frac{\eta}{R} (e^{R(T-t)} -1)p}
e^{\frac{1}{4 R} \eta^2\sigma^2_0 (e^{2R(T-t)} -1)}.
\end{split} 
\end{equation}
Defining
\begin{equation}
\label{eqn:h}
h(t) \doteq \eta e^{R(T-t)} \bigl( \frac{1}{2} \eta e^{R(T-t)} \sigma_0^2 - p \bigr),
\end{equation}
we can write
\begin{equation}
\label{g}
g(t,x) =   e^{-\eta x e^{R (T-t)}} e^{\int_t^T h(s) ds}.
\end{equation}
Hence, using \eqref{eqn:Psi} and \eqref{ex}, inequality \eqref{eqn:V<g} reads as
\begin{multline}
\int_t^T [\Psi(s) - h(s)]\,ds \\
= \frac{\eta q}{R} (e^{R(T-t)}-1) -  \frac{1}{2}\frac{q^2}{\sigma_0^2}(T-t) - \frac{1}{4 R} \eta^2\sigma^2_0 (e^{2R(T-t)} -1)\leq 0 \qquad  \forall t \in [0,T]. 
\end{multline}
\end{Remark}


 \section{Reduction to an optimal stopping problem}
 \label{sec:red}
 
We can show that the mixed stochastic control problem \eqref{V} can be reduced to an optimal stopping problem. Let us denote by $\mathcal{T}_{t,T}$ is the set of $\mathbb{F}$-stopping times $\tau$ such that $t \leq \tau \leq T$.

\begin{Theorem} \label{reduction}
We have that
\begin{equation}\label{stop}
V(t,x) =\inf_{\tau\in\mathcal{T}_{t,T}}{\mathbb{E}\bigl[ \bar V(\tau, X^{t,x}_\tau - K ) I_{\{\tau < T\}} + e^{-\eta X^{t,x}_\tau} I_{\{\tau = T\}} \bigr]}
\quad  \forall  t \in[0,T] \times \mathbb{R},
\end{equation}
where $\bar{V}$ is given in \eqref{V bar} and $X^{t,x} = \{X^{t,x}_s \}_{s\in [t,T]}$ denotes the wealth process given in equation \eqref{W1}, with initial data $(t,x) \in [0,T] \times \mathbb{R}$.

Moreover, let $\tau^*_{t,x} \in \mathcal{T}_{t,T}$  an optimal stopping time for problem \eqref{stop}. Then 
$\alpha^* = ( \tau^*_{t,x}, \{u^*_s\}_{s \in [\tau^*_{t,x},T]} )$, with  $u^* = \{ u^*_t\}_{t\in[0,T]}$ given in  \eqref{u_star},  is an optimal strategy for problem \eqref{V},
with the convention that on the event $\{ \tau^*_{t,x} =T\}$ we take $\alpha^* = ( T, 1)$.

\end{Theorem}

\begin{proof}
We first prove the inequality 
$$V(t,x) \geq \inf_{\tau\in\mathcal{T}_{t,T}}{\mathbb{E}\bigl[ \bar V(\tau, X^{t,x}_\tau - K ) I_{\{\tau < T\}} +  e^{-\eta X^{t,x}_\tau} I_{\{\tau = T\}} \bigr]} . $$
For any arbitrary strategy $ \alpha = ( \tau, \{u_s\}_{s \in [\tau,T]} ) \in \mathcal{A}_t$, we have that  $\tau \in \mathcal{T}_{t,T}$  and
 by  \eqref{wealth_alpha}
 
 $${\mathbb{E}\bigl[ e^{-\eta X^{\alpha, t,x} _T} \bigr ]} =     
 {\mathbb{E} \bigl[ e^{-\eta X^{u, \tau, X^{t,x}_\tau - K }_T} I_{\{\tau < T\}}  +   e^{-\eta X^{t,x}_\tau}  I_{\{\tau = T\}}  \bigr ]}   =$$
 
 $$ {\mathbb{E} \bigl[ {\mathbb{E}\bigl [ e^{-\eta X^{u, \tau, X^{t,x}_\tau - K }_T}  |  \mathcal{F}_\tau \bigr ]} I_{\{\tau < T\}}  +   e^{-\eta X^{t,x}_\tau}  I_{\{\tau = T\}}  \bigr ]}  \geq $$
  
 $$ {\mathbb{E}\bigl[ \bar V(\tau, X^{t,x}_\tau - K)  I_{\{\tau < T\}}  +  e^{-\eta X^{t,x}_\tau}  I_{\{\tau = T\}} \bigr]} .$$ 
 
Taking the infimum over $\alpha  \in\mathcal{A}_t$  on both side leads to the desired inequality.
The other side of the inequality is based on the fact that there exists $u^* = \{ u^*_t\}_{t\in[0,T]} \in \mathcal{U}$, given in \eqref{u_star} optimal for the problem \eqref{V bar}. Indeed, consider the strategy $\bar \alpha = ( \tau, \{u^*_s\}_{s \in [\tau,T]} ) \in\mathcal{A}_t$ where $\tau$ is  arbitrary 
 chosen in $\mathcal{T}_{t,T}$.  Then  
 $$V(t,x) \leq \mathbb{E}\bigl[ e^{-\eta X^{ \bar \alpha, t,x} _T} \bigr] = \mathbb{E} \bigl[ e^{-\eta X^{u^*, \tau, X^{t,x}_\tau - K }_T} I_{\{\tau < T\}}  +   e^{-\eta X^{t,x}_\tau}  I_{\{\tau = T\}}  \bigr ]  = $$
 $$\mathbb{E}\bigl[ \bar V(\tau, X^{t,x}_\tau - K )   I_{\{\tau < T\}}  +   e^{-\eta X^{t,x}_\tau}  I_{\{\tau = T\}} \bigr].$$

Taking the infimum over $ \tau \in \mathcal{T}_{t,T}$ on the right-hand side gives that 
$$V(t,x) \leq \inf_{\tau\in\mathcal{T}_{t,T}}{\mathbb{E}\bigl[ \bar V(\tau, X^{t,x}_\tau - K ) I_{\{\tau < T\}} + e^{-\eta X^{t,x}_\tau} I_{\{\tau = T\}} \bigr]}
\quad  \forall  t \in[0,T] \times \mathbb{R},$$
and hence the equality \eqref{stop}.

Finally, let $\tau^*_{t,x} \in \mathcal{T}_{t,T}$  an optimal stopping time for problem \eqref{stop} and  
$\alpha^* = ( \tau^*_{t,x}, \{u^*_s\}_{s \in [\tau^*_{t,x},T]} )$, with  $u^* = \{ u^*_t\}_{t\in[0,T]}$ given in  \eqref{u_star}, we get that 
\[
\begin{split}
& V(t,x) = \mathbb{E}\bigl[ \bar V(\tau^*_{t,x}, X^{t,x}_{\tau^*_{t,x}} - K ) I_{\{\tau^*_{t,x}< T\}} + e^{-\eta X^{t,x}_{\tau^*_{t,x}}} I_{\{\tau^*_{t,x} = T\}} \bigr] \\
&	= \mathbb{E} \bigl[ e^{-\eta X^{u^*, \tau^*_{t,x}, X^{t,x}_{\tau^*_{t,x}} - K }_T} I_{\{{\tau^*_{t,x}} < T\}} 
	+   e^{-\eta X^{t,x}_{\tau^*_{t,x}}}  I_{\{{\tau^*_{t,x}} = T\}}  \bigr ] \\
&	= \mathbb{E}\bigl[ e^{-\eta X^{ \alpha^*, t,x} _T} \bigr] ,
\end{split}
\]
and this concludes the proof.
\end{proof}

According to Theorem \ref{reduction} we can solve the original problem given in \eqref{V} in two steps: after investigating the pure reinsurance problem \eqref{V bar} (see Section \ref{sec:reins}), we can analyze the optimal stopping problem \eqref{stop}, which is the main goal of the next section.


\section{The  optimal stopping problem}
\label{sec:optstop}
 
In this section we discuss the optimal stopping problem \eqref{stop}:
$$
V(t,x) =\inf_{\tau\in\mathcal{T}_{t,T}}{\mathbb{E}\bigl[ \bar V(\tau, X^{t,x}_\tau - K ) I_{\{\tau < T\}} + e^{-\eta X^{t,x}_\tau} I_{\{\tau = T\}} \bigr]}
\qquad  t \in[0,T] \times \mathbb{R} . $$

Let us observe that
$$V(T,x)  =  e^{-\eta x} \qquad \forall x \in \mathbb{R},$$
while choosing $\tau=t < T$ and $\tau=T$ in the right hand side of \eqref{stop}, we get that
\begin{equation} \label{disu}
V(t,x)  \leq \bar V(t,x-K) \qquad \text{and} \quad V(t,x)  \leq \mathbb{E}\bigl[   e^{-\eta X^{t,x}_T} \bigr] =g(t,x), \forall (t,x)\in[0,T]\times \mathbb{R},
\end{equation}
respectively.\\

Now denote by $\mathcal{L}$  the Markov generator of the process $X^{t,x}$:
\begin{equation}
\label{eqn:L_X}
\mathcal{L} f(t,x)  = \frac{\partial{f}}{\partial{t}}(t,x) 
+ \bigl(Rx + p \bigr) \frac{\partial{f}}{\partial{x}}(t,x) 
+ \frac{1}{2} \sigma_0^2 \frac{\partial^2{f}}{\partial{x^2}}(t,x),
\end{equation}
with $f \in \mathcal{C}^{1,2}((0,T)\times\mathbb{R})$.

\begin{Remark}
From the theory of optimal stopping (see, for instance \cite{oksendal:sde}), when the cost function  $G(t,x)$ is continuous and the value function 
$$W(t,x) =\inf_{\tau\in\mathcal{T}_{t,T}}{\mathbb{E}\bigl[ G(\tau, X^{t,x}_\tau)]},
 \qquad  t \in[0,T] \times \mathbb{R} $$
is sufficiently regular, it can be characterized as a solution to the following variational inequality:
\begin{equation}\label{var} \min\{ \mathcal{L} W(t,x), G(t,x) - W(t,x)\}=0, \quad (t,x) \in (0,T) \times \mathbb{R}. \end{equation}
This is a free-boundary problem, whose solution is the function $W(t,x)$ and the so-called continuation region, which is defined as
\begin{equation}
\mathcal{C} = \{ (t,x) \in (0,T) \times \mathbb{R} : W(t,x) < G(t,x) \}.
\end{equation}
Moreover,  it is known that the first exit time of the process $X^{t,x}$ from the region $\mathcal{C}$   
 \[
\tau^*_{t,x} \doteq \inf \{ s \in [t,T] :  (s,X^{t,x}_s)  \notin   \mathcal{C} \}.
\]
provides an optimal stopping time.

In our optimal stopping problem \eqref{stop}, the cost function is
$$\bar V(t, x - K ) I_{\{t < T\}} + e^{-\eta x} I_{\{t = T\}}, \quad t \in[0,T] \times \mathbb{R}, $$
which is not continuous on $[0,T] \times \mathbb{R}$, hence the classical theory on optimal
stopping  problems does not directly apply. 
\end{Remark}
 
In view of the preceding remark, we now prove a Verification Theorem which applies to our specific problem.


\begin{Theorem}[Verification Theorem] \label{verification}
Let $\varphi \colon [0,T]\times\mathbb{R}\to\mathbb{R}$ be a function satisfying the assumptions below and
$\mathcal{C}$ (the continuation region) be defined by
\begin{equation}
\label{eqn:C_reg}
\mathcal{C} = \{ (t,x) \in (0,T) \times \mathbb{R} :  \varphi(t,x) < \bar V(t,x-K) \}.
\end{equation}
Suppose that  the following conditions are satisfied.
\begin{enumerate}
\item There exists $t^* \in [0,T)$ such that $\mathcal{C}=(t^*,T)\times\mathbb{R}$. 

\item $\varphi \in \mathcal{C}([0,T]\times\mathbb{R})$, $\varphi$ is  $\mathcal{C}^{1}$ w.r.t $t$ in $(0,t^*)$ and $(t^*, T)$, separately, and $\mathcal{C}^{2}$ w.r.t. $x \in\mathbb{R}$;

\item $\varphi(t,x) \leq \bar V(t,x-K)$ $\forall  (t,x)\in [0,T] \times \mathbb{R}$ and
$\varphi(T,x) = e^{-\eta x}$ $\forall x\in\mathbb{R}$;

\item $\varphi$ is a solution to the following variational inequality 
\begin{equation}\label{var1}
\begin{cases}
\mathcal{L} \varphi(t,x) \geq 0 	&\forall (t,x) \in (0,t^*) \times \mathbb{R} \\
 \mathcal{L} \varphi(t,x)=0 		&\forall (t,x) \in \mathcal{C} = (t^*, T) \times \mathbb{R} .
\end{cases}
 \end{equation}
 

\item the family $\{ \varphi(\tau, X_\tau); \tau \in  \mathcal{T}_{0,T} \}$  is uniformly integrable.
\end{enumerate}
Moreover,  let $\tau^*_{t,x}$ the first exit time from the region $\mathcal{C}$ of the process $X^{t,x}$, that is  
 \[
\tau^*_{t,x} \doteq \inf \{ s \in [t,T] :  (s,X^{t,x}_s)  \notin   \mathcal{C} \}.
\]
with the convention  $\tau^*_{t,x} =T$ if the set on the right-hand side is empty.\\
Then $\varphi(t,x) = V(t,x)$ on $[0,T] \times \mathbb{R}$ and $ \tau^*_{t,x} $  is an optimal stopping time for problem \eqref{stop}.
\end{Theorem}
\begin{proof}


For any $(t,x)\in [0,T) \times \mathbb{R}$ let us take the sequence of stopping times $\{\tau_n\}_{n\ge1}$ such that $\tau_n\doteq\inf\{s\ge t\mid |X^{t,x}_s| \ge n\}$. 
We first prove that, $\forall \tau \in \mathcal{T}_{t,T}$

\begin{equation} \label{First}
\varphi(t,x) \leq \mathbb{E}[\varphi(\tau\land \tau_n,X^{t,x}_{\tau\land \tau_n})], \quad \forall (t,x)\in [0,T) \times \mathbb{R}.
\end{equation}

Due to the specific form of the continuation region we have two cases.
If $t \geq t^*$, since $\varphi \in \mathcal{C}^{1,2}((t^*,T)\times\mathbb{R})$, 
applying Dynkin's formula%
\footnote{Notice that we use a localization argument, so that $\tau\land \tau_n$ is the first exit time of a bounded set and, as s consequence, $\varphi$ is not required to have a compact support (see \cite[Theorem 7.4.1]{oksendal:sde}).}
 we get that for any arbitrary stopping time $\tau \in \mathcal{T}_{t,T}$ 
\[
\varphi(t,x) = \mathbb{E}[\varphi(\tau\land \tau_n,X^{t,x}_{\tau\land \tau_n})] 
- \mathbb{E}\biggl[ \int_t^{\tau\land \tau_n} \mathcal{L} \varphi(s,X^{t,x}_s)\,ds \biggr] =  \mathbb{E}[\varphi(\tau\land \tau_n,X^{t,x}_{\tau\land \tau_n})] .
\]

If  $t  < t^*$, we have again by Dynkin's formula, since  $\varphi \in \mathcal{C}^{1,2}((0,t^*)\times\mathbb{R})$, that
\[
\begin{split}
 \varphi(t,x) &= \mathbb{E}[\varphi(\tau\land \tau_n\land t^*,X^{t,x}_{\tau\land \tau_n\land t^*})] 
- \mathbb{E}\biggl[ \int_t^{\tau\land \tau_n \land t^*} \mathcal{L} \varphi(s,X^{t,x}_s)\,ds \biggr] \\
& \leq \mathbb{E}[\varphi(\tau\land \tau_n\land t^*,X^{t,x}_{\tau\land \tau_n\land t^*})] 
\end{split}
\]
and, similarly, since $\varphi \in \mathcal{C}^{1,2}((t^*,T)\times\mathbb{R})$,
 \[
 \begin{split}
\mathbb{E}[\varphi(\tau\land \tau_n\land t^*,X^{t,x}_{\tau\land \tau_n\land t^*})]  
 & = \mathbb{E}[\varphi(\tau\land \tau_n,X^{t,x}_{\tau\land \tau_n})] 
- \mathbb{E}\biggl[ \int_{\tau\land \tau_n \land t^*}^{\tau\land \tau_n} \mathcal{L} \varphi(s,X^{t,x}_s)\,ds \biggr]
\\
& = \mathbb{E}[\varphi(\tau\land \tau_n,X^{t,x}_{\tau\land \tau_n})],  
\end{split}
 \]
hence \eqref{First} is proved.

Now letting $n\to+\infty$ in \eqref{First}, recalling that $\varphi  \in \mathcal{C}([0,T]\times\mathbb{R})$ and using  Fatou Lemma  
 we get that
\[
\begin{split}
\varphi(t,x) \leq \mathbb{E}[\varphi(\tau,X^{t,x}_\tau)] \leq \mathbb{E}[\bar V(\tau,X^{t,x}_\tau-K)I_{\{\tau < T\}} + e^{-\eta X^{t,x}_\tau} I_{\{\tau = T\}}]  \quad \forall \tau \in \mathcal{T}_{t,T},
\end{split}
\]
 hence $\varphi(t,x) \leq V(t,x)$, $\forall (t,x)\in [0,T) \times \mathbb{R}$. 
To prove the opposite inequality we  consider four different cases.
\begin{enumerate}
\item  If the stopping region is not empty, that is $t^* \in (0,T)$,  $\forall (t,x)\in (0, t^*) \times \mathbb{R}$ we know that $\varphi(t,x) = \bar{V}(t,x-K) \geq V(t,x)$,  hence $\varphi(t,x) = V(t,x)$, which implies $V(t,x) = \bar{V}(t,x-K)$ and $\tau^*_{t,x}=t$ is optimal for problem \eqref{stop}. 
\item If the stopping region is not empty,  for $t=0$, we have that  $\varphi(0,x) = \bar V(0,x-K)$ $\forall x \in \mathbb{R}$, 
otherwise by continuity of both the functions if $\varphi(0,x) > \bar V(0,x-K)$ (or $\varphi(0,x) < \bar V(0,x-K)$) the same inequality holds in a neighborhood of $(0,x)$ which contradicts  that 
$\varphi(t,x) = \bar V(t,x-K)$, $\forall (t,x) \in (0,t^*)\times \mathbb{R}$.  Then $\varphi(0,x) = V(0,x)$ $\forall x \in \mathbb{R}$ and $\tau^*_{0,x}=0$ is optimal for problem \eqref{stop}. 
\item If the continuation region is not empty, that is $t^* \in [0,T)$, $\forall (t,x)\in  [t^*,T) \times \mathbb{R}$, repeating the localization argument with the stopping time $\tau^*_{t,x} =T$, we get 
\[
\begin{split}
\varphi(t,x) &= \mathbb{E}[\varphi(T,X^{t,x}_T)] 
- \mathbb{E}\biggl[ \int_t^T \mathcal{L} \varphi(s,X^{t,x}_s)\,ds \biggr] \\
&= \mathbb{E}[\varphi (T,X^{t,x}_T)] = \mathbb{E}[ e^{-\eta X^{t,x}_{T} }] \ge V(t,x),
\end{split}
\]
as a consequence $\varphi(t,x) = V(t,x)=\mathbb{E}[ e^{-\eta X^{t,x}_{T} }]$ and  $\tau^*_{t,x} = T$  is optimal for problem \eqref{stop}.
 \item Finally, for $t=T$ by assumption $\varphi(T,x) =  e^{-\eta x} = V(T,x)$, $\forall x \in   \mathbb{R}$,  $\tau^*_{T,x} = T$ is optimal for problem \eqref{stop} 
and this concludes the proof.
\end{enumerate}
\end{proof}

\begin{Lemma}\label{unif}
Let $g$ as defined in equation \eqref{g}. The families $\{ \bar V(\tau, X_\tau-K); \tau \in  \mathcal{T}_{0,T} \}$  and  $\{  g(\tau,X_\tau); \tau \in  \mathcal{T}_{0,T} \}$
are uniformly integrable.
\end{Lemma}
\begin{proof}
Recalling that $\bar V(t,x) \leq g(t,x)$ by \eqref{eqn:V<g}, we have that $\bar V(t, x-k) \leq e ^{\eta K e^{RT}} g(t,x)$, hence the statement follows by the uniformly integrability of the family 
$$\{ g(\tau,X_\tau):  \tau \in  \mathcal{T}_{0,T} \}.$$ 
It is well known that if for any arbitrary $\delta>0$ and any stopping time $\tau \in \mathcal{T}_{0,T}$  
\[
 \mathbb{E}[ g(\tau,X_\tau)^{1+\delta}]  < +\infty,
 \]
then the proof is complete. To this end, we observe that
\[
\begin{split}
\mathbb{E}[ g(\tau,X_\tau)^{1+\delta}] 
	&= \mathbb{E}[ e^{(1 + \delta) \int_\tau^T h(s) ds } e^{- \eta (1 + \delta) e^{R(T- \tau )} X_\tau }] \\
	&\le e^{\frac{1}{4 R} (1 + \delta)  \eta^2\sigma^2_0 e^{2RT} }
	e^{-\eta (1 + \delta) R_0 e^{RT}}\mathbb{E}[e ^{- \eta (1 + \delta) \int_0^\tau e^{R(T- s )} \sigma_0 dW_s}]\\
	&\le e^{\frac{1}{4 R} (1 + \delta)  \eta^2\sigma^2_0 e^{2RT} }
	e^{-\eta (1 + \delta) R_0 e^{RT}}  e^{\frac{1}{4 R} (1 + \delta)^2  \eta^2 \sigma^2_0 (e^{2RT} -1) } <+\infty.
\end{split}
\]
\end{proof}
%
%

The  guess for the continuation region $ \mathcal{C}$ given in the assumption $1.$  of the Verification Theorem follows by the next result.

\begin{Lemma}
\label{lemma:A}
The set
\begin{equation}
\label{eqn:setA}
A \doteq \{(t,x) \in (0,T) \times \mathbb{R} :  \mathcal{L} \bar V(t, x-K) < 0 \} 
\end{equation}
is included in the continuation region, that is 
$$A \subseteq  \mathcal{C} = \{ (t,x) \in (0,T) \times \mathbb{R} :  V(t,x) < \bar V(t,x-K) \}.$$
Moreover,  the following equation holds:
\[
A = (t_A,T)\times\mathbb{R} ,
\]
where
\begin{equation}
t_A \doteq 0\lor 
\biggl[T- \frac{1}{R} \log{\biggl( \frac{q+RK + \sqrt{(q+RK)^2 - q^2}}{\eta\sigma_0^2} \biggr)} \biggr]
\land T .
\end{equation}
In particular, only three cases are possible, depending on the model parameters:
\begin{enumerate}
\item if 
\[
y^*:=\frac{q+RK + \sqrt{(q+RK)^2 - q^2}}{\eta\sigma_0^2} \ge e^{RT},
\]
then $t_A = 0$ and $\mathcal{L} \bar V(t, x-K) < 0$ $\forall (t,x)\in (0,T)\times\mathbb{R}$, so that $A= (0,T) \times\mathbb{R}$, implying that $\mathcal{C}= (0,T) \times\mathbb{R}$;
\item if 
\[
1 <  y^* = \frac{q+RK + \sqrt{(q+RK)^2 - q^2}}{\eta\sigma_0^2} < e^{RT},
\]
then $0 < t_A<T$ and $\mathcal{L} \bar V(t, x-K) < 0$ $\forall (t,x)\in(t_A,T)\times\mathbb{R}$; in this case $A=(t_A,T)\times\mathbb{R}$;
\item if 
\[
y^*= \frac{q+RK + \sqrt{(q+RK)^2 - q^2}}{\eta\sigma_0^2} \le 1,
\]
then $t_A = T$ and $\mathcal{L} \bar V(t, x-K) \ge 0$ $\forall (t,x)\in (0,T)\times\mathbb{R}$, so that $A=\emptyset$.
\end{enumerate}
\end{Lemma}
\begin{proof}
 First let us observe that $\bar V(t, x -K) \in \mathcal{C}^{1,2}((0,T)\times\mathbb{R})\cap\mathcal{C}([0,T]\times\mathbb{R})$ and the family $\{ \bar V(\tau, X_\tau-K); \tau \in  \mathcal{T}_{0,T} \}$  is uniformly integrable by Lemma \ref{unif}.
Now choose $(\bar t, \bar x) \in A$,  let $B \subset A$ be a neighborhood of $(\bar t, \bar x)$ with $\tau_B<T$, where $\tau_B$ denotes the first exit time of $X^{\bar t, \bar x}$ from $B$. 
 Then by Dynkin's formula  
 \[
\begin{split}
\bar V(\bar t, \bar x -K) &= \mathbb{E}[\bar V(\tau_B,X^{\bar t, \bar x}_{\tau_B} - K)] 
- \mathbb{E}\biggl[ \int_{\bar t}^ {\tau_B}\mathcal{L} \bar V(s,X^{\bar t, \bar x}_s - K)\,ds \biggr] \\
&> \mathbb{E}[\bar V(\tau_B,X^{\bar t, \bar x}_{\tau_B} - K)] \geq V(\bar t, \bar x).
\end{split}
\]
Hence $(\bar t, \bar x) \in \mathcal{C}$ and $A\subseteq\mathcal{C}$.


Next, recalling \eqref{eqn:L_X}, we have that 
\[
\mathcal{L} \bar V(t, x-K) = \bar V(t, x-K) \bigl( - \Psi(t)
- \eta e^{R(T-t)} (p+ RK)
+ \frac{1}{2} \eta^2 e^{2R(T-t)} \sigma_0^2 \bigr) ,
\]
so that $\mathcal{L} \bar V(t, x-K) < 0$ if and only if 
\begin{equation}
\label{eqn:Psi_ineq}
\Psi(t) > \frac{1}{2} \eta^2 e^{2R(T-t)} \sigma_0^2
- \eta e^{R(T-t)} (p+ RK) ,
\end{equation}
that is, using \eqref{eqn:Psi},
\[
\frac{1}{2} \eta^2 e^{2R(T-t)} \sigma_0^2 - \eta e^{R(T-t)} (q+ RK) +  \frac{1}{2}\frac{q^2}{\sigma_0^2} <0.
\]
Using a change of variable $z=e^{R(T-t)}$, we can rewrite the inequality as
\[
\frac{1}{2} \eta^2 \sigma_0^2 z^2 - \eta (q+ RK) z +  \frac{1}{2}\frac{q^2}{\sigma_0^2} < 0.
\]
Since $\eta^2 [(q+K)^2 - q^2]>0$ the associated equation admits two different solutions, so that the inequality \eqref{eqn:Psi_ineq} is satisfied by
\[
\frac{q+ RK - \sqrt{(q+ RK)^2 - q^2}}{\eta\sigma_0^2} < z < \frac{q+ RK + \sqrt{(q+ RK)^2 - q^2}}{\eta\sigma_0^2} .
\]
Recalling \eqref{eqn:qsmall}, we can verify that
\[
\frac{q+ RK - \sqrt{(q+ RK)^2 - q^2}}{\eta\sigma_0^2} < \frac{q}{\eta\sigma_0^2} < 1 ,
\]
so that the inequality reads as
\[
t_A = T- \frac{1}{R} \log{\biggl( \frac{q+ RK + \sqrt{(q+ RK)^2 - q^2}}{\eta\sigma_0^2} \biggr)} < t < T .
\]
Depending on the model parameters, we can see that only the three cases above are possible. Equivalently, $\mathcal{L} \bar V(t, x-K) < 0$ if and only if $t_A<t<T$.
\end{proof}

\begin{Remark}\label{forse}
As consequence of Lemma  \ref{lemma:A}, recalling \eqref{eqn:h}, in \textit{Cases 1} and \textit{2}, that is when $0\leq t_A <T$, we have that 
$$\Psi(t) - h(t) + \eta RK e^{R(T-t)}>0, \quad \forall t > t_A ,$$
see equation \eqref{eqn:Psi_ineq}, which implies, $\forall t \geq t_A$
\[
\int_t^T (\Psi(s) - h(s) ) ds  + \int_t^T \eta R K e^{R(T-s)} ds > 0,
\]
equivalently 
\[
\int_t^T (\Psi(s) - h(s) ) ds  +  \eta  K e^{R(T-t)} > \eta  K
\]
for all $t \in [t_A, T)$.  In \textit{Case 3}, that is when $t_A=T$, since 
$$\Psi(t) - h(t) + \eta RK e^{R(T-t)} <0, \quad \forall t \in [0,T)$$
we have that 
\[
\int_t^T (\Psi(s) - h(s) ) ds + \eta  K e^{R(T-t)} < \eta  K,
\]
for all $t \in [0, T)$.
\end{Remark}

We need the following preliminary result to provide an explicit expression for the value function of the problem \eqref{stop}.

\begin{Lemma}
\label{lemma:Vtilde}
The function $\widetilde{V}(t, x) = Cg(t,x)$, $(t,x)\in (0,T)\times\mathbb{R}$, with $C$ any positive constant and $g$ as given in equation \eqref{g}, is a solution to the PDE 
$\mathcal{L} \widetilde{V}(t, x) = 0$, $(t,x)\in (0,T)\times\mathbb{R}$.\\
In particular, $g$ is a solution to the PDE with boundary condition $g(T,x)=e^{-\eta x}$ $\forall x\in\mathbb{R}$.
\end{Lemma}
\begin{proof}
Using the ansatz $\widetilde{V}(t, x) = e^{-\eta x e^{R(T-t)}} \gamma(t)$, we can reduce the PDE $\mathcal{L} \widetilde{V}(t, x) = 0$ to the following equation:
\[
e^{-\eta xe^{R(T-t)}} \gamma'(t) 
- \eta e^{R(T-t)} V(t, x) p
+ \frac{1}{2} \eta^2 e^{2R(T-t)} V(t, x) \sigma_0^2 = 0 ,
\]
which is equivalent to this ODE:
\[
\gamma'(t) + h(t)\gamma(t) = 0, \quad (t,x)\in (0,T)\times\mathbb{R},
\]
where the function $h$ is given in \eqref{eqn:h}.

Since the solution of the ODE is $\gamma(t)= C\, e^{\int_t^T h(s)\,ds}$, we get the expression of $\widetilde{V}$ as above.\\
Finally, setting $C=1$, $g$ satisfies the PDE above with the terminal condition $g(T,x)=e^{-\eta x}$ $\forall x\in\mathbb{R}$.
\end{proof}


Before proving the main result of this section, which is Theorem \ref{thm:auxstop}, we compare $g(t,x)$, given in \eqref{g}, with $\bar V(t,x-K)$.  

\begin{Lemma}\label{xi}
Let 
\begin{equation}\label {H}
H(t) = \int_t^T ( \Psi(s) - h(s) ) ds + \eta K e^{R(T-t)}, \qquad t \in [0,T] ,
\end{equation}
then we distinguish two cases:
\begin{enumerate}
\item if $H(0) \geq 0$, then $g(t,x)  < \bar V(t,x-K)$ $ \forall (t,x) \in (0,T]\times\mathbb{R}$;
\item if $H(0) < 0$, then there exists $t^* \in (0,t_A)$ such that $g(t^*, x)  = \bar V(t^*, x-K)$ $\forall x \in \mathbb{R}$ and $g(t,x)  < \bar V(t,x-K)$ $\forall (t,x) \in (t^*,T]\times\mathbb{R}$.
\end{enumerate}
\end{Lemma}
\begin{proof}
Let us observe that the  inequality  $g(t,x) <   \bar V(t,x-K)$ writes as
\[
e^{-\eta x e^{R(T-t)}}e^{\int_t^T h(s)\,ds} < e^{-\eta (x-K) e^{R(T-t)} } e^{\int_t^T \Psi(s) ds},
\]
that is
\[
e^{\int_t^T (\Psi(s) - h(s))\,ds}e^{\eta K e^{R(T-t)} }  > 1
\Leftrightarrow  \int_t^T (\Psi(s) - h(s))\,ds + \eta K e^{R(T-t)} = H(t) > 0. 
\]
We distinguish three cases:  
\begin{itemize}
\item [(i)] when $0\leq t_A <T$, we have that $H(t) \geq \eta K >0$ $\forall t > t_A $ by Remark \ref{forse} and it easy to verify that $H$ is increasing in $[0,t_A]$, while it is decreasing in $[t_A,T]$. Hence,  it takes the maximum value at $t=t_A$. As a consequence, if $H(0) \geq 0$ we have that $H(t)>0$ $\forall t \in (0,T]$, being $H(T)=\eta K>0$.

Otherwise, if $H(0) < 0$ there exists $t^* \in (0,t_A)$ such that $H(t^*) =0$, that is $g(t^*, x)  = \bar V(t^*, x-K)$ $\forall x \in \mathbb{R}$, and $H(t) >0$ $\forall (t,x) \in (t^*,T]$, that is $g(t,x)  < \bar V(t,x-K)$ $\forall (t,x) \in (t^*,T]\times\mathbb{R}$;
\item [(ii)]  when $t_A =T$, by Lemma \ref{lemma:A} we get that $H$ is increasing in $[0,T]$ and we can repeat the same arguments as in the previous case to distinguish the two casese $H(0) \geq 0$ and $H(0) <0$, obtaining the same results;
\item[(iii)] when $t_A=0$, by Remark \ref{forse}  we know that $H$ is decreasing in $[0,T]$, so that $H(t) \geq \eta K >0$ $\forall t \in [0,T]$, that is $g(t,x)  < \bar V(t,x-K)$, $ \forall (t,x) \in (0,T]\times\mathbb{R}$. Moreover, in this case $H(0)\geq 0$.
\end{itemize} 
Summaring, we obtain our statement.
\end{proof}

We now prove some properties of the continuation region.

\begin{Proposition}\label{PropA}
Let 
 \begin{equation}
\mathcal{C} = \{ (t,x) \in (0,T) \times \mathbb{R} :  V(t,x) < \bar V(t,x-K) \}.
\end{equation}
Then 
we distinguish two cases:
\begin{enumerate}
\item if $H(0)\geq 0$, then $\mathcal{C} = (0,T) \times \mathbb R$,
\item if $H(0) < 0$, then $ (t^*,T) \times \mathbb R \subseteq \mathcal{C}$, where $t^* \in (0, t_A)$ is the unique solution to equation
$$H(t) = \int_t^T ( \Psi(s) - h(s) ) ds + \eta K e^{R(T-t)}=0.$$
\end{enumerate}
\end{Proposition}
\begin{proof}

%

We apply Lemma \ref{xi}.  
In Case 1,  we have that $V(t,x) \leq g(t,x) < \bar V(t, x-K)$ $\forall (t,x) \in (0,T)\times\mathbb{R}$, that is 
$\mathcal{C} = (0,T) \times \mathbb R$.
In Case 2,  we have that $V(t,x) \leq g(t,x) < \bar V(t, x-K)$ $\forall (t,x) \in (t^*,T)\times\mathbb{R}$,  which implies $ (t^*,T) \times \mathbb R \subseteq \mathcal{C}$,
and this concludes the proof.
\end{proof}

Now we are ready for the main result of this section.

\begin{Theorem}
\label{thm:auxstop}
Let $H$ be given in \eqref{H}.
The solution of the optimal stopping problem \eqref{stop} takes different forms, depending on the model parameters. Precisely, we have two cases:
\begin{enumerate}
\item if  $H(0) =  \int_0^T ( \psi(s) - h(s) ) ds + \eta K e^{RT} \geq 0$,
then the continuation region is $\mathcal{C} = (0,T) \times \mathbb R$, the value function is
\[
V(t,x) = g(t,x) = e^{-\eta x e^{R(T-t)}} e^{\int_t^T h(s)\,ds}, \qquad (t,x)\in [0,T] \times\mathbb{R}
 \]
and  $\tau^*_{t,x} = T$ is an optimal stopping time;
\item if $H(0)=  \int_0^T ( \psi(s) - h(s) ) ds + \eta K e^{RT}  < 0$, then $\mathcal{C} = (t^*, T) \times \mathbb R$, where $t^* \in (t_A,T)$ is the unique solution to $H(t)=0$, the value function is    
\begin{equation}
V(t,x) = 
\begin{cases}
 \bar V(t, x-K)= e^{-\eta (x-K) e^{R(T-t)} } e^{\int_t^T \Psi(s) ds} 
& (t,x)\in[0, t^*]\times\mathbb{R} \\
g(t,x)=   \mathbb{E}\bigl[ e^{- \eta X^{t,x}_T}] =  e^{-\eta x e^{R(T-t)}} e^{\int_t^T h(s)\,ds} & (t,x)\in (t^*,T]\times\mathbb{R} 
\end{cases}
\end{equation}
and $\tau^*_{t,x}$, given by
\begin{equation}
\label{eqn:tau*}
\tau^*_{t,x} = 
\begin{cases}
t 	& (t,x)\in[0,t^*]\times\mathbb{R} \\
T 	& (t,x)\in(t^*,T]\times\mathbb{R} ,
\end{cases}
\end{equation}
is an optimal stopping time.
\end{enumerate}
\end{Theorem}

\begin{proof}
We prove the two cases separately, applying Theorem \ref{verification} in each one.\\

\textit{Case 1}\\
The continuation region is $\mathcal{C} = (0,T) \times \mathbb R$ by Proposition \ref{PropA}, hence assumption 1 of Theorem  \ref{verification} is fulfilled. Moreover, $\tau^*_{t,x} = T$. Observing that 
$$g(t,x) = e^{-\eta x e^{R(T-t)}} e^{\int_t^T h(s)\,ds} \in \mathcal{C}^{1,2}((0,T)\times\mathbb{R})\cap\mathcal{C}([0,T]\times\mathbb{R}),$$
the assumption 2 of Theorem  \ref{verification} is clearly matched. The assumption 3 is implied by Lemma \ref{xi}. Moreover,  the variational inequality \eqref{var1} (assumption 4) is fulfilled by Lemma \ref{lemma:Vtilde}.
Finally, by Lemma \ref{unif} the last condition in Theorem \ref{verification} is fulfilled.

\textit{Case 2}\\
$\mathcal{C} = (t^*, T) \times \mathbb R$ clearly satisfies the first assumption of Theorem \ref{verification}. Taking
\begin{equation*}
\varphi(t,x) = 
\begin{cases}
 \bar V(t, x-K)= e^{-\eta (x-K) e^{R(T-t)} } e^{\int_t^T \Psi(s) ds} 
& (t,x)\in[0, t^*]\times\mathbb{R} \\
g(t,x)=   \mathbb{E}\bigl[ e^{- \eta X^{t,x}_T}] =  e^{-\eta x e^{R(T-t)}} e^{\int_t^T h(s)\,ds} & (t,x)\in (t^*,T]\times\mathbb{R} ,
\end{cases}
\end{equation*}
observing that Lemma \ref{xi} ensures the existence of $t^* \in (0,t_A)$ such that $g(t^*,x) = \bar V(t^*, x-K)$ when $H(0) < 0$, the smoothness conditions of the second assumption are matched.
Moreover, according to Lemma \ref{xi}, $g(t,x) < \bar V(t, x-K)$ $\forall (t,x) \in (t^*, T]$ and the assumption 3 is fulfilled.
That the variational inequality \eqref{var1} is satisfied by $\varphi$ is a consequence of the results of Section \ref{sec:reins} and of Lemma \ref{lemma:Vtilde}. Finally, Lemma \ref{unif} implies the fifth assumption of Theorem \ref{verification} and the proof is complete. 
\end{proof}

\section{Solution to the original problem}
\label{sec:solution}

As a direct consequence of the results obtained in the previous section and Theorem \ref{reduction}, we provide an explicit solution to the optimal reinsurance problem under 
fixed cost given in \eqref{V}. 

\begin{Theorem}
\label{prop:pb_original_sol}
Let us define
\begin{equation}
\label{eqn:K*}
K^* =  -\frac{q}{R} (1-e^{-RT}) +  \frac{1}{2}\frac{q^2}{\eta \sigma_0^2}T e^{- RT} 
+ \frac{1}{4 R} \eta\sigma^2_0 (e^{RT} -e^{- RT}) > 0.
\end{equation}
Two cases are possible, depending on the model parameters:
\begin{enumerate}
\item 
if  $K\ge K^*$, then the value function given in \eqref{V} is
\[
V(t,x) = g(t,x) = \mathbb{E}\bigl[ e^{- \eta X^{t,x}_T}] 
\]
and the optimal strategy is $\alpha^*=(T,1)$, that is no reinsurance is purchased;
\item if $K< K^*$, then the value function is
\[
V(t,x) = 
\begin{cases}
\bar V(t, x-K)= e^{-\eta (x-K) e^{R(T-t)} } e^{\int_t^T \Psi(s) ds} 
& (t,x)\in[0, t^*]\times\mathbb{R} \\
g(t,x)=   \mathbb{E}\bigl[ e^{- \eta X^{t,x}_T}] =  e^{-\eta x e^{R(T-t)}} e^{\int_t^T h(s)\,ds} 
& (t,x)\in (t^*,T]\times\mathbb{R} ,
\end{cases}
\]
where $t^* \in (0,T)$ is the unique solution to the equation
\[
\eta (\frac{q}{R} +  K) e^{R(T-t)}  -  \frac{1}{2}\frac{q^2}{\sigma_0^2}(T-t) - \frac{1}{4 R} \eta^2\sigma^2_0 (e^{2R(T-t)} -1) -  \frac{\eta q}{R}=0,
\]
and the optimal strategy is $\alpha^*=(\tau^*_{t,x},\{ \frac{q}{\eta\sigma_0^2}e^{-R(T-s)} \}_{s\in[\tau^*_{t,x},T]})$, with $\tau^*_{t,x}$ given in \eqref{eqn:tau*}.
 \end{enumerate}
\end{Theorem}

\begin{proof}
 Let us observe that, using Remark \ref{idea}, 
\[
\begin{split}
H(t) =  & \int_t^T ( \Psi(s) - h(s) ) ds +  \eta K e^{R(T-t)} \\
= & \frac{\eta q}{R} (e^{R(T-t)}-1) -  \frac{1}{2}\frac{q^2}{\sigma_0^2}(T-t) - \frac{1}{4 R} \eta^2\sigma^2_0 (e^{2R(T-t)} -1) +  \eta K e^{R(T-t)}.
\end{split} 
\]
and the condition $H(0) \geq 0$ is equivalent to
\[ 
K \geq -  \frac{1}{\eta} e^{- RT} \int_0^T ( \Psi(s) - h(s) ) ds = K^*,
\]
while the condition $H(0) < 0$ can be written as $K< K^*$. That $K^*>0$ follows by Remark \ref{idea}.
Then the statement is a consequence of Theorem \ref{thm:auxstop}.
  \end{proof}
 
Let us briefly comment the two cases of Theorem \ref{prop:pb_original_sol}. Case 1 corresponds to no reinsurance. That is, the insurer is not willing to subscribe a contract at any time of the selected time horizon. Besides the insurer, this result is relevant for the reinsurance company. We have proven that there exists a threshold $K^*>0$ (see equation \eqref{eqn:K*}), which represents the maximum initial cost that the insurer is willing to pay to buy reinsurance. If the reinsurer chooses a subscription cost higher than $K^*$, then the insurer will not buy protection from her.

In Case 2, at any time $t\in[0,T]$, the insurer immediately subscribes the reinsurance agreement if the time instant $t^*$ has not passed, applying the optimal retention level from that moment on; otherwise, if $t>t^*$, no reinsurance will be bought.

We notice that it is never optimal to wait for buying reinsurance. That is, it is convenient either to immediately sign the contract, or not to subscribe at all.

In particular, at the starting time $t=0$, given an initial wealth $R_0>0$, we have these cases:
\begin{enumerate}
\item if $K\ge K^*$, then $\alpha^*=(T,1)$, that is no reinsurance is purchased;
\item if $K< K^*$, then $\alpha^*=(0,\{ \frac{q}{\eta\sigma_0^2}e^{-R(T-s)} \}_{s\in[0,T]}$, that is
the optimal  choice for the insurer consists in stipulating the contract at the initial time, selecting the optimal retention level (as in the pure reinsurance problem).
\end{enumerate}

By the expression \eqref{eqn:K*} we can show that $K^*$ is increasing with respect to $\eta$ and $\sigma_0$, while it is decreasing with respect to $q$. More details will be given in the next section by means of numerical simulations.\\

Another relevant result for the reinsurance company is the following.

\begin{Proposition}
For any fixed cost $K>0$ there exists $q^*\in(0,+\infty)$ (depending on $K$) such that
\begin{enumerate}
\item 
if  $q>q^*$, then
\[
V(t,x) = g(t,x) = \mathbb{E}\bigl[ e^{- \eta X^{t,x}_T}] 
\]
and $\alpha^*=(T,1)$, that is no reinsurance is purchased;
\item otherwise
\[
V(t,x) = 
\begin{cases}
\bar V(t, x-K)= e^{-\eta (x-K) e^{R(T-t)} } e^{\int_t^T \Psi(s) ds} 
& (t,x)\in[0, t^*]\times\mathbb{R} \\
g(t,x)=   \mathbb{E}\bigl[ e^{- \eta X^{t,x}_T}] =  e^{-\eta x e^{R(T-t)}} e^{\int_t^T h(s)\,ds} 
& (t,x)\in (t^*,T]\times\mathbb{R} ,
\end{cases}
\]
where $t^* \in (0,T)$ is the unique solution to the equation
\[
\eta (\frac{q}{R} +  K) e^{R(T-t)}  -  \frac{1}{2}\frac{q^2}{\sigma_0^2}(T-t) - \frac{1}{4 R} \eta^2\sigma^2_0 (e^{2R(T-t)} -1) -  \frac{\eta q}{R}=0,
\]
and $\alpha^*=(\tau^*_{t,x},\{ \frac{q}{\eta\sigma_0^2}e^{-R(T-s)} \}_{s\in[\tau^*_{t,x},T]})$, with $\tau^*_{t,x}$ is given in \eqref{eqn:tau*}.
 \end{enumerate}
\end{Proposition}
\begin{proof}
Following Theorem \ref{prop:pb_original_sol} and its proof, we can write the condition $H(0)\le0$ as
\[
\frac{T}{2\sigma_0^2}q^2 + (1-e^{RT})\frac{\eta}{R} q 
	+ \frac{\eta^2\sigma^2_0}{4 R}  (e^{2RT} -1) - \eta K e^{RT} \ge 0 .
\]
To simplify our computations, let us consider this inequality for any $q\in\mathbb{R}$. The discriminant $\Delta$ must be positive, otherwise the existence of $K^*>0$ in Theorem \ref{prop:pb_original_sol} is not guaranteed anymore. The solutions of the associated equations are
\[
q_{1,2} = \frac{\eta\sigma_0^2}{T}\bigl( \frac{e^{RT}-1}{R} \pm \sqrt{\Delta} \bigr), \qquad q_1<q_2.
\]
Since
\[
q_2 > \frac{\eta\sigma_0^2(e^{RT}-1)}{RT} > \eta\sigma_0^2,
\]
only $q_1$ is relevant because of the condition \eqref{eqn:qsmall}. That $q_1\in(0,+\infty)$ is a consequence of the existence of $K^*>0$ in Theorem \ref{prop:pb_original_sol}. If $q_1$ was not positive, then $H(0)>0$ for any value of $q>0$ and this would contradict Theorem \ref{prop:pb_original_sol}. Setting $q^*=q_1$ concludes the proof.
\end{proof}

The last result is interesting for the reinsurer. In Section \ref{sec:reins} we have already stated that the condition $q<\eta\sigma_0^2$ (see equation \eqref{eqn:qsmall}) is required in order that the reinsurance agreement is desirable. In presence of a fixed initial cost, now we know that there exists a threshold $q^*$, which is smaller than $\eta\sigma_0^2$, such that the insurer will never subscribe the contract if $q>q^*$.

\begin{Remark}
Recalling that $q = \theta\lambda\mu$ (see Section \ref{section:formulation}, we can give a deeper interpretation of the previous result. Indeed, we have proven the existence of a maximum safety loading $\theta^*>0$, which cannot be exceeded by the reinsurer, otherwise the reinsurance contract will not be subscribed.
\end{Remark}


\section{Numerical simulations}
\label{sec:num}

In this section we use some numerical simulations in order to further investigate the results obtained in Section \ref{sec:solution}. Unless otherwise specified, all the simulations are performed according to the parameters of Table 1 below.

\begin{table}[h]
\centering
\begin{tabular}{lc}
\toprule
Parameter & Value \\
\midrule
$T$ 		& $10$ \\
$\eta$		& $0.5$ \\
$\sigma_0$	& $0.5$ \\
$q$			& $0.1$ \\
$R$			& $0.05$ \\
\bottomrule
\end{tabular}
\caption{Model parameters.}
\end{table}

\begin{figure}[h]
\centering
\subfloat[][\emph{Reinsurer's net profit}]
{\includegraphics[width=.45\textwidth]{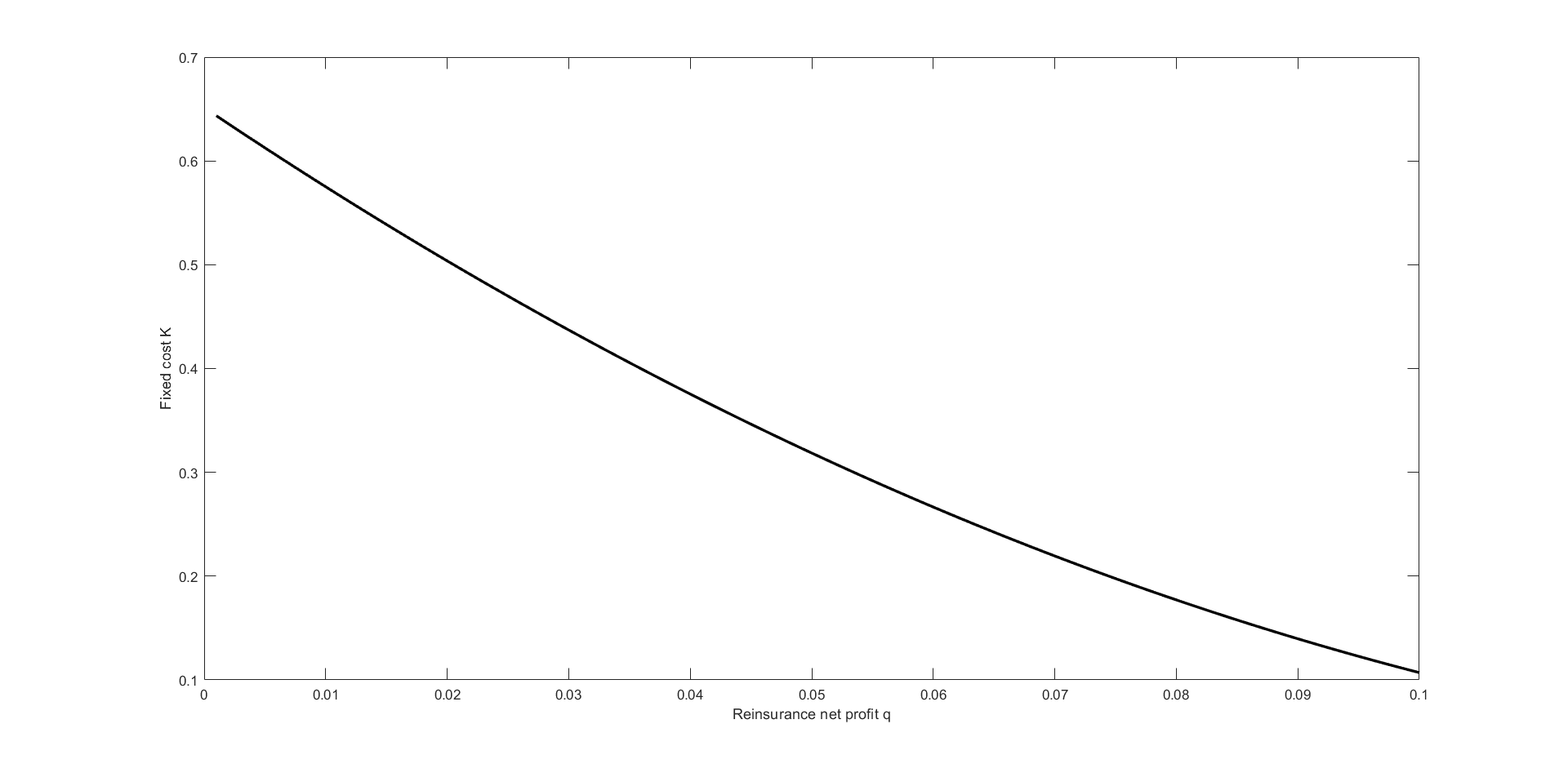}} \quad
\subfloat[][\emph{Risk aversion}]
{\includegraphics[width=.45\textwidth]{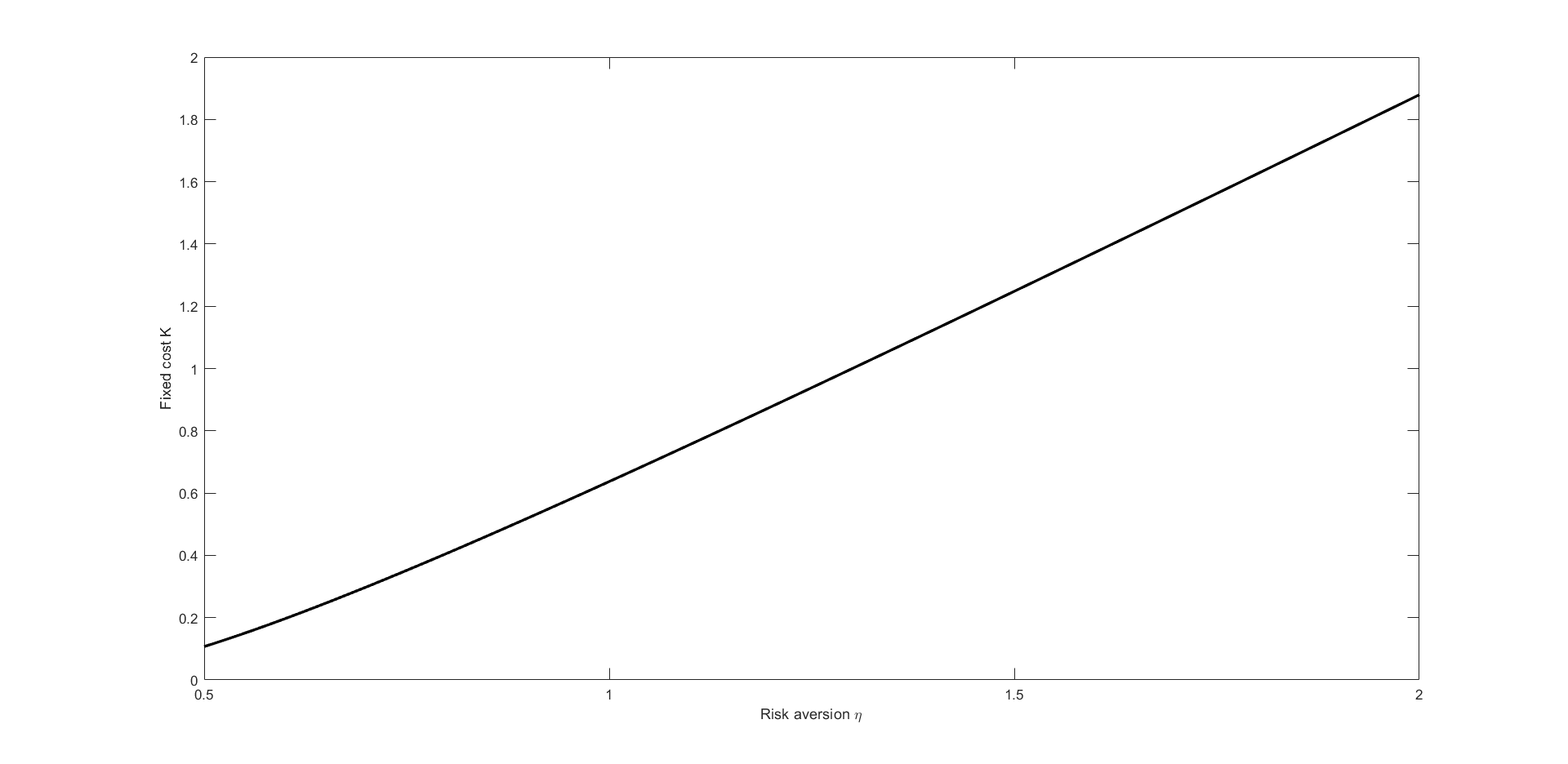}} \\
\subfloat[][\emph{Volatility}]
{\includegraphics[width=.45\textwidth]{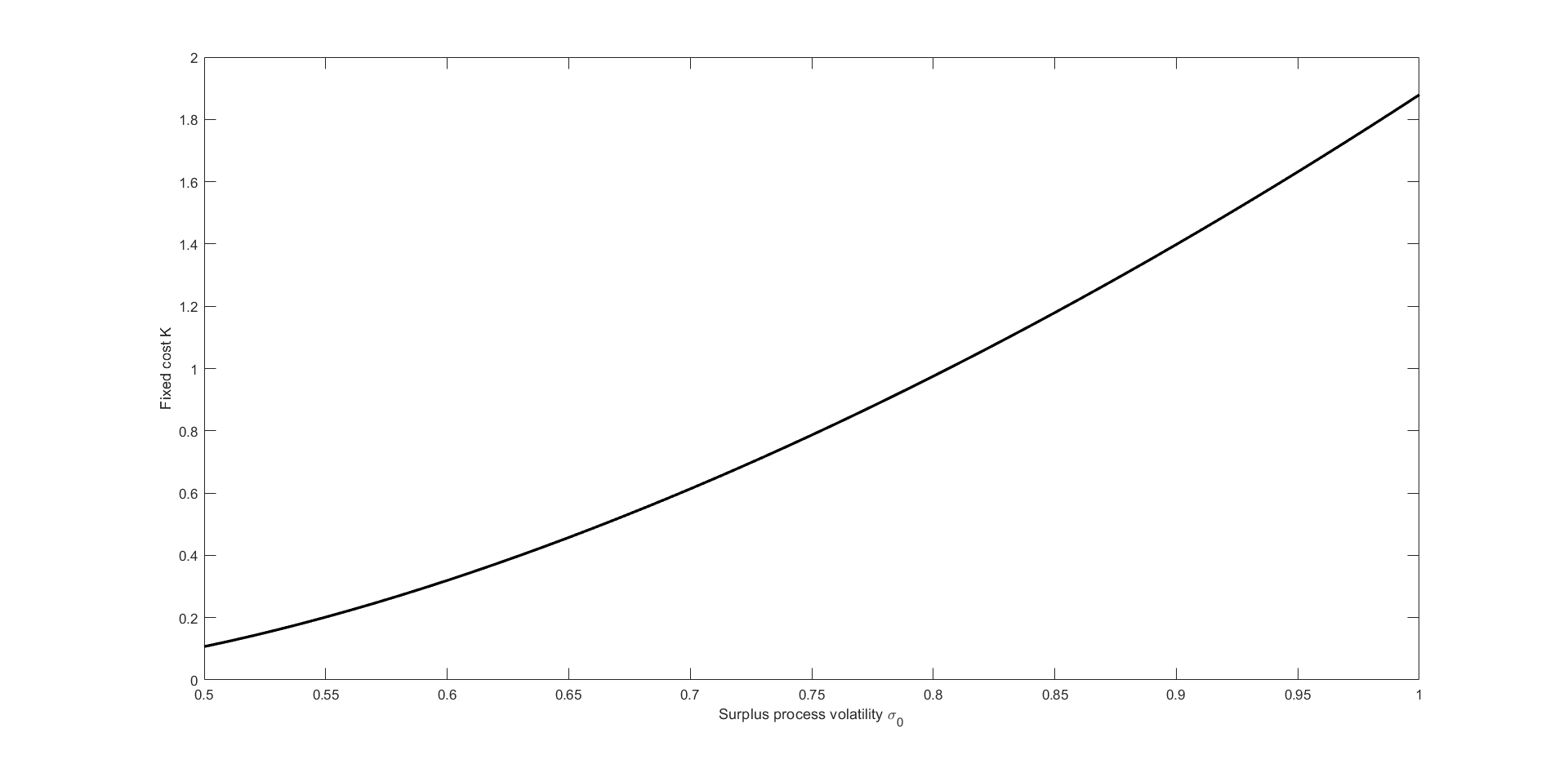}} \quad
\subfloat[][\emph{Time horizon}]
{\includegraphics[width=.45\textwidth]{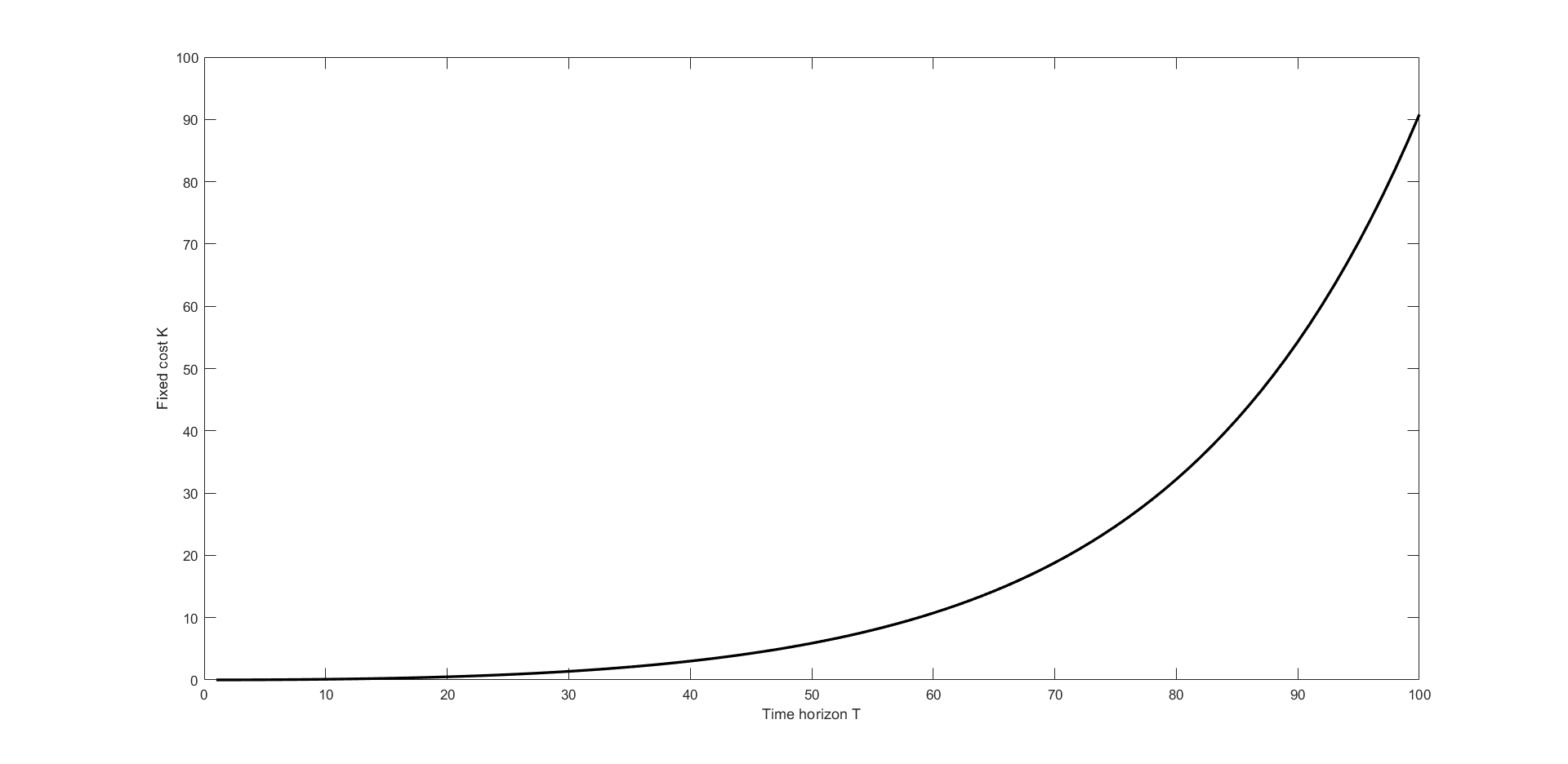}}
\caption{Sensitivity analyses with respect to the model parameters.}
\label{fig:K*}
\end{figure}

We have previously illustrated how the threshold $K^*$ in equation \eqref{eqn:K*} is relevant for the insurer as well as for the reinsurer. Figure \ref{fig:K*} shows how this threshold is influenced by the model parameters. As expected, if the reinsurer increases her net profit (for example, increasing the safety loading), then the fixed cost should decrease, see Figure \ref{fig:K*}a. When the insurer is more risk averse, she is willing to pay a higher fixed cost, see Figure \ref{fig:K*}b. The same applies when the potential losses increase, that is $\sigma_0$ is high, as in Figure \ref{fig:K*}c. Figure \ref{fig:K*}d shows that the more the time horizon, the higher will be $K^*$.

\newpage
\bibliographystyle{apalike}
\bibliography{biblio}

\end{document}